\documentclass[a4paper]{article}
\usepackage{geometry}
\usepackage{amsmath}
\usepackage{amssymb}
\usepackage{amsthm}
\usepackage{mathrsfs}
\usepackage{bm}
\usepackage{cases}
\usepackage{braket}

\theoremstyle{plain}
\newtheorem{lemma}{Lemma}
\newtheorem{theorem}{Theorem}

\let\origeqref\eqref
\newcommand{\myeqref}[1]{\origeqref{eq:#1}}
\let\eqref\myeqref

\DeclareMathOperator{\GL}{GL}
\DeclareMathOperator{\LGL}{LGL}
\DeclareMathOperator{\linspan}{span}
\DeclareMathOperator{\tr}{tr}
\DeclareMathOperator{\diag}{diag}
\DeclareMathOperator{\cof}{cof}
\DeclareMathOperator{\codim}{codim}

\DeclareMathOperator{\vdim}{vdim}
\DeclareMathOperator{\eval}{ev}
\DeclareMathOperator{\Hom}{Hom}
\DeclareMathOperator{\End}{End}
\newcommand{\deq}{\mathrel{:=}}
\newcommand{\isom}{\cong}
\newcommand{\bb}[1]{\mathbb{#1}}
\newcommand{\e}{\mathrm{e}}
\newcommand{\app}[3]{#1\colon #2\to #3}
\newcommand{\rist}[2]{\left. #1\right|_{#2}}
\newcommand{\quoz}[2]{{#1/{#2}}}
\newcommand{\pair}[2]{\langle#1,#2\rangle}
\newcommand{\comm}[2]{[#1,#2]}
\newcommand{\abs}[1]{\left|#1\right|}
\newcommand{\trasp}[1]{{#1}^{t}}

\newcommand{\oGr}{\mathbb{G}\mathrm{r}}
\newcommand{\Gr}{\mathrm{Gr}}
\newcommand{\oGrrat}{\oGr^{\mathrm{rat}}}
\newcommand{\Grrat}{\Gr^{\mathrm{rat}}}

\newcommand{\Grhom}{\Gr^{\mathrm{hom}}}
\newcommand{\Grad}{\Gr^{\mathrm{ad}}}

\newcommand{\CM}{\mathcal{C}}
\newcommand{\CP}{\bb{C}P^{1}}

\begin{document}

\title{On rational solutions of multicomponent and matrix KP hierarchies}
\author{Alberto Tacchella\thanks{E-mail: \texttt{tacchell@dima.unige.it}}}
\date{\today}

\maketitle

\begin{abstract}
  We derive some rational solutions for the multicomponent and matrix KP
  hierarchies generalising an approach by Wilson. Connections with the
  multicomponent version of the KP/CM correspondence are discussed.
\end{abstract}

\section{Introduction}

The \emph{KP hierarchy} is an integrable hierarchy of partial differential
equations generated by a pseudo-differential operator of the form
\begin{equation}
  \label{eq:1}
  L = D + u_{1}D^{-1} + u_{2}D^{-2} + \dots
\end{equation}
Here the \((u_{i})_{i\geq 1}\) are elements of a differential algebra
\(\mathcal{A}\) of smooth functions of a variable \(x\) (with
\(D=\partial/\partial x\)) and a further infinite family of variables \(\bm{t}
= (t_{i})_{i\geq 1}\) . The evolution of \(L\) is determined by the following
system of Lax-type equations:
\begin{equation}
  \label{eq:2}
  \partial_{k}L = \comm{B_{k}}{L} \qquad (k\geq 1)
\end{equation}
where \(\partial_{k} = \partial/\partial t_{k}\) and \(B_{k} = (L^{k})_{+}\). 
All these equations commute, and the variable \(x\) may be identified with
\(t_{1}\). This hierarchy of equations is naturally viewed as a dynamical
system defined on an infinite-dimensional Grassmann manifold, as discovered by
Sato \cite{sato81}. In \cite{sw85} Segal and Wilson developed a very general
framework for building solutions to the KP hierarchy out of the points of a
certain Grassmannian \(\Gr(H)\) of closed subspaces in a Hilbert space \(H\).

The \emph{multicomponent KP hierarchy} (mcKP) is a generalisation of the KP
hierarchy obtained by replacing \(\mathcal{A}\) with the differential algebra
of \(r\times r\) matrices whose entries belong to an algebra of smooth
functions of a variable \(x\) and \(r\) further families of variables
\[ \bm{t}^{(1)} = (t^{(1)}_{i})_{i\geq 1} \quad \dots \quad \bm{t}^{(r)} =
(t^{(r)}_{i})_{i\geq 1} \]
which we will collectively denote by \(\bar{\bm{t}}\). The operator \eqref{1}
now becomes
\begin{equation}
  \label{eq:3}
  L = D + U_{1} D^{-1} + U_{2} D^{-2} + \dots
\end{equation}
where \((U_{i})_{i\geq 1}\) are \(r\times r\) matrices. To define the
evolution of \(L\) we introduce a set of \(r\) matrix pseudo-differential
operators \(R_{\alpha}\) of the form\footnote{Greek indices will henceforth
  run from \(1\) to \(r\).}
\begin{equation}
  \label{eq:4}
  R_{\alpha} = E_{\alpha} + R_{1\alpha} D^{-1} + R_{2\alpha} D^{-2} + \dots
\end{equation}
where \(E_{\alpha}\) is the matrix with \(1\) in the entry \((\alpha,\alpha)\)
and zero elsewhere. These operators are required to satisfy the equations
\(\comm{L}{R_{\alpha}} = 0\), \(\comm{R_{\alpha}}{R_{\beta}} = 0\) and
\(\sum_{\alpha} R_{\alpha} = I_{r}\) (it can be shown that such operators do
exist). The evolution equations for the mcKP hierarchy are then
\begin{equation}
  \label{eq:5}
  \partial_{k\alpha} L = \comm{B_{k\alpha}}{L} \qquad (k\geq 1,\,
  \alpha=1\dots r)
\end{equation}
where \(\partial_{k\alpha} = \partial/\partial t_{k}^{(\alpha)}\) and
\(B_{k\alpha} = (L^{k}R_{\alpha})_{+}\). The variable \(x\) may be identified
with \(\sum_{\gamma} t_{1}^{(\gamma)}\); if we define also for each \(k\geq
2\) the new variables \(t_{k} = \sum_{\gamma} t_{k}^{(\gamma)}\) then the
corresponding flows determine a sub-hierarchy of equations that we call the
\emph{matrix KP hierarchy}.

In \cite{wils93} George Wilson obtained a complete classification of the
rational solutions to the KP hierarchy. More precisely, he proved that the
coefficients of an operator \(L\) satisfying \eqref{1} are proper (i.e.,
vanishing at infinity) rational functions of \(x\) exactly when \(L\) comes
from a point of a certain sub-Grassmannian of \(\Gr(H)\), which he called the
\emph{adelic Grassmannian}. 
% Moreover, in \cite{wils98} it is proved that this
% space is in one-to-one correspondence with the phase space of the rational
% Calogero-Moser system, thereby providing a geometric explanation for the
% phenomenon (first noticed by Krichever in \cite{krich78}) that the poles of a
% rational solution to the KP equation evolve as point particles subject to the
% Calogero-Moser Hamiltonian. In turn, both 
This space turns out to be in one-to-one
correspondence with the phase space of the rational Calogero-Moser system
\cite{wils98}; this provides a geometric explanation for the phenomenon, first
noticed by Krichever in \cite{krich78}, that the poles of a rational solution
to the KP equation evolve as a system of point particles described by the
Calogero-Moser Hamiltonian. On the other hand, the adelic Grassmannian may
also be seen as the moduli space of isomorphism classes of right ideals in the
Weyl algebra \cite{bw00}, thereby linking the subject to the emerging field of
noncommutative algebraic geometry.
This web of connections %between spaces of such a disparate origin
is sometimes referred to as the ``KP/CM correspondence'' (see also \cite{bzn03}
for a wider perspective on that matter).

The purpose of this paper is to establish some rationality results, obtained
in the author's PhD Thesis \cite{tesi}, for the solutions of multicomponent
and matrix KP hierarchies. Our main motivation is to understand how the
above-mentioned results generalise to the multicomponent setting. The paper is
organised as follows. In section \ref{sec:2} we briefly recall the mapping
between points in the \(r\)-component Segal-Wilson Grassmannian \(\Gr(r)\) and
solutions of multicomponent KP. In section \ref{sec:3} we analyse the
multicomponent rational Grassmannian \(\Grrat(r)\) and display an explicit
formula for the Baker and tau functions associated to these points. In section
\ref{sec:4} we prove two rationality results: in Theorem \ref{th:mcKP} we
consider mcKP solutions coming from certain (very special) points in
\(\Grrat(r)\), whereas in Theorem \ref{th:mKP} we restrict to the matrix KP
hierarchy but consider a much larger subset of \(\Grrat(r)\). Finally in
section \ref{sec:5} we briefly comment on the relevance of these results for
the multicomponent version of the KP/CM correspondence.

\paragraph{Acknowledgements.}

The author would like to thank George Wilson, Igor Mencattini, Claudio
Bartocci and Volodya Rubtsov for helpful discussions related to this work.

\section{Preliminaries}
\label{sec:2}

We start by briefly recalling the definition of the \(r\)-component
Segal-Wilson Grassmannian \cite{sw85,ps86}. Consider the Hilbert space
\(H^{(r)} = L^{2}(S^{1},\bb{C}^{r})\); its elements can be thought of as
functions \(\bb{C}\to \bb{C}^{r}\) by embedding \(S^{1}\) in the complex plane
as the circle \(\gamma_{R}\) with centre \(0\) and radius \(R\in \bb{R}^{+}\).
We have the splitting \(H^{(r)} = H^{(r)}_{+}\oplus H^{(r)}_{-}\) in the two
subspaces consisting of functions with only positive (resp. negative) Fourier
coefficients, with associated orthogonal projections \(\pi_{\pm}\). The
\emph{Segal-Wilson Grassmannian} of \(H^{(r)}\), denoted by \(\Gr(r)\), is the
set of all closed linear subspaces \(W\subseteq H^{(r)}\) such that
\(\rist{\pi_{+}}{W}\) is a Fredholm operator of index zero and
\(\rist{\pi_{-}}{W}\) is a compact operator.

If we take the elements of \(H^{(r)}\) to be row vectors for definiteness, the
loop group \(\LGL(r,\bb{C})\) naturally acts on \(H^{(r)}\) (hence on
\(\Gr(r)\)) by matrix multiplication from the right. We define
\(\Gamma_{+}(r)\) as the subgroup consisting of diagonal matrices of the form%
\footnote{More generally, we could consider the subgroup determined by a
  maximal torus of type \(\underline{r}\) (where \(\underline{r}\) is any
  partition of \(r\)) in \(\GL(r,\bb{C}\)); this gives rise to the so-called
  \emph{Heisenberg flow of type} \(\underline{r}\). However, as shown e.g. in
  \cite{dm95}, these flows are simply the pullback on \(\Gr(r)\) of multiple
  lower-dimensional mcKP flows.}
\begin{equation}
  \label{eq:6}
  \diag(g_{1}, \dots, g_{r}) \quad\text{ with } g_{\alpha}\in \Gamma_{+}
\end{equation}
where \(\Gamma_{+}\) is the group of analytic functions \(\app{g}{S^{1}}
{\bb{C}^{*}}\) that extend to holomorphic functions on the disc \(\set{z\in
  \CP | \abs{z}\leq R}\) and such that \(g(0)=1\) (cfr. \cite{sw85}). It
follows that for each \(g_{\alpha}\) there exists a holomorphic function
\(f_{\alpha}\) such that \(g_{\alpha} = \e^{f_{\alpha}}\) (with
\(f_{\alpha}(0) = 0\)), and by letting \(f_{\alpha} = \sum_{i\geq 1}
t^{(\alpha)}_{i} z^{i}\) we have
\begin{equation}
  \label{eq:7}
  g_{\alpha}(z) = \exp \sum_{i\geq 1} t^{(\alpha)}_{i} z^{i}
\end{equation}
Thus a generic matrix \(g\in \Gamma_{+}(r)\) may be written in the form
\begin{equation}
  \label{eq:8}
  g = \exp \diag(\sum_{i\geq 1} t^{(1)}_{i} z^{i}, \dots, \sum_{i\geq 1}
  t^{(r)}_{i} z^{i})
\end{equation}
and is totally described by the family of coefficients \(\bar{\bm{t}} =
(\bm{t}^{(1)}, \dots, \bm{t}^{(r)})\). 

We now recall the mapping between points of \(\Gr(r)\) and solutions to the
multicomponent KP hierarchy. In what follows we will say that a matrix-valued
function \(\psi(z)\) belongs to a subspace \(W\in \Gr(r)\) if and only if each
row of \(\psi\), seen as an element of \(H^{(r)}\), belongs to \(W\).

For every \(W\in \Gr(r)\) we define
\begin{equation}
  \label{eq:9}
  \Gamma_{+}(r)^{W}\deq \set{g\in \Gamma_{+}(r) | Wg^{-1} \text{ is
      transverse}}
\end{equation}
where ``transverse'' means that the orthogonal projection \(Wg^{-1}\to
H^{(r)}_{+}\) is an isomorphism. For any \(g\in \Gamma_{+}(r)^{W}\) the
\emph{reduced Baker function} associated to \(W\) and \(g\) is the
matrix-valued function \(\psi\) whose row \(\psi_{\alpha}\) is the inverse
image of \(e_{\alpha}\in H^{(r)}_{+}\) (the \(\alpha\)-th element of the
canonical basis of \(\bb{C}^{r}\)) by \(\rist{\pi_{+}}{Wg^{-1}}\). It follows
straightforwardly that
\begin{equation}
  \label{eq:10}
  \tilde{\psi}_{W}(g,z) = I_{r} + \sum_{i\geq 1} W_{i}(g) z^{-i}
\end{equation}
for some matrices \((W_{i})_{i\geq 1}\). Now, since each row of the matrix
\(\tilde{\psi}_{W}\) belongs to the subspace \(Wg^{-1}\), each row of the
product matrix \(\tilde{\psi}_{W}g\) will belong to \(W\); the \emph{Baker
  function} associated to \(W\) is the map \(\psi_{W}\) which sends \(g\in
\Gamma_{+}(r)^{W}\) to this matrix:
\begin{equation}
  \label{eq:11}
  \psi_{W}(g,z) = \biggl( I_{r} + \sum_{i\geq 1} W_{i}(g) z^{-i}\biggr) g(z)
\end{equation}
Notice that for every \(\eta\in \LGL(r,\bb{C})\) one has
% \footnote{For the sake
%   of completeness we give a quick proof of this formula: by definition,
%   \[ \psi_{W\eta}(g,z) = \rist{\pi_{+}}{W\eta g^{-1}}^{-1}(1) g \]
%   now put \(h\deq g\eta^{-1}\), then \(h^{-1} = \eta g^{-1}\) so
%   \[ = \rist{\pi_{+}}{Wh^{-1}}^{-1}(1) h\eta = \psi_{W}(h,z)\eta \]
%   as claimed.}
\(\tilde{\psi}_{W\eta}(g,z) = \tilde{\psi}_{W}(g\eta^{-1},z)\), so that
\begin{equation}
  \label{eq:14}
  \psi_{W\eta}(g,z) = \psi_{W}(g\eta^{-1},z)\cdot \eta
\end{equation}
We now recall that expressions such as \eqref{11} are in one-to-one
correspondence with zeroth-order pseudo-differential operators of the form
\begin{equation}
  \label{eq:12}
  K_{W} = I_{r} + \sum_{i\geq 1} W_{i}(\bar{\bm{t}}) D^{-i}
\end{equation}
where we consider \(g\) as a function of the coefficients \(\bar{\bm{t}}\)
defined by \eqref{8}. From \(K_{W}\) we can define a first order
pseudo-differential operator \(L_{W}\) via the following prescription
(``dressing''):
\begin{equation}
  \label{eq:13}
  L_{W} = K_{W} D K_{W}^{-1}
\end{equation}
The following result was proved in \cite{dick93}:
\begin{theorem}
  For any point \(W\in \Gr(r)\) the operator \(L_{W}\) defined by \eqref{13}
  satisfies the multicomponent KP equation.
\end{theorem}
We remark that the correspondence \eqref{13} is not one-to-one: if \(K' = KC\)
where \(C = I_{r} + \sum_{j\geq 1} C_{j} D^{-j}\) for some family
\((C_{j})_{j\geq 1}\) of constant diagonal matrices then \(L=L'\). At the
level of \(\Gr(r)\) this ``gauge freedom'' is expressed by the action of the
group \(\Gamma_{-}(r)\) consisting of diagonal \(r\times r\) matrices of the
form \(\diag(h_{1}, \dots, h_{r})\), where each \(h_{\alpha}\) belongs to the
group (denoted \(\Gamma_{-}\) in \cite{sw85}) of analytic functions
\(\app{h}{S^{1}}{\bb{C}^{*}}\) that extend to holomorphic functions on the
disc \(\set{z\in \CP | \abs{z}\geq R}\) and such that \(h(\infty)=1\).

Another way to describe the Baker function associated to a point of \(\Gr(r)\)
relies on the so-called \emph{tau function}. Here we do not need to enter into
the details of its definition (see again \cite{dick93}); the key fact is that
to each subspace \(W\in \Gr(r)\) we can associate certain holomorphic
functions on \(\Gamma_{+}(r)\), denoted \(\tau_{W}\) and
\(\tau_{W\alpha\beta}\) for each pair of indices \(\alpha\neq \beta\),
determined up to constant factors, such that the following equality (``Sato's
formula'') holds:
\begin{equation}
  \label{eq:15}
  \tilde{\psi}_{W}(g,z)_{\alpha\beta} =
  \begin{cases}
    \displaystyle
    \frac{\tau_{W}(gq_{z\alpha})}{\tau_{W}(g)} &\text{ if } \alpha=\beta\\
    \displaystyle
    z^{-1} \frac{\tau_{W\alpha\beta}(gq_{z\beta})}{\tau_{W}(g)} &\text{ if }
    \alpha\neq\beta
  \end{cases}
\end{equation}
where for each \(z\in \bb{C}\) and \(\alpha=1,\dots,r\) we define
\(q_{z\alpha}\) to be the element of \(\Gamma_{+}(r)\) that has
\(q_{z}(\zeta)\deq 1 - \frac{\zeta}{z}\) at the \((\alpha,\alpha)\) entry
and \(1\) elsewhere on the diagonal.

\section{The multicomponent rational Grassmannian}
\label{sec:3}

\subsection{Definition}

For the sake of brevity we will denote by \(\mathcal{R}\) the space of
rational functions on the complex projective line \(\CP\), by \(\mathcal{P}\)
the subspace of polynomials and by \(\mathcal{R}_{-}\) the subspace of proper
(i.e., vanishing at infinity) rational functions. We consider the space
\(\mathcal{R}^{r}\) of \(r\)-tuples of rational functions with the direct sum
decomposition
\begin{equation}
  \label{eq:dec-Rm}
  \mathcal{R}^{r} = \mathcal{P}^{r}\oplus \mathcal{R}_{-}^{r}
\end{equation}
and associated canonical projection maps \(\app{\pi_{+}}{\mathcal{R}^{r}}
{\mathcal{P}^{r}}\) and \(\app{\pi_{-}}{\mathcal{R}^{r}}{\mathcal{R}_{-}^{r}}\).

We define the Grassmannian \(\oGrrat(r)\) as the set of closed linear
subspaces \(W\subseteq \mathcal{R}^{r}\) for which there exist polynomials
\(p,q\in \mathcal{P}\) such that
\begin{equation}
  \label{eq:def-oGrrat-m}
  p\mathcal{P}^{r}\subseteq W\subseteq q^{-1}\mathcal{P}^{r}
\end{equation}
The \emph{virtual dimension} of \(W\), denoted \(\vdim W\), is the index of
the (Fredholm) operator \(p_{+}\deq \rist{\pi_{+}}{W}\); one has
\begin{equation}
  \label{eq:vdim-W}
  \vdim W = \dim W' - r \deg p
\end{equation}
where \(W'\deq \quoz{W}{p\mathcal{P}^{r}}\) is finite-dimensional. We denote
by \(\Grrat(r)\) the subset of \(\oGrrat(r)\) consisting of subspaces of
virtual dimension zero.

This space can be embedded in the \(r\)-component Segal-Wilson Grassmannian
\(\Gr(r)\) by the following procedure: given \(W\in \Grrat(r)\), we choose the
radius \(R\in \bb{R}^{+}\) involved in the definition of \(\Gr(r)\) such that
every root of the polynomial \(q\) appearing in \eqref{def-oGrrat-m} is
contained in the open disc \(\abs{z} < R\); then the restrictions
\(\rist{f}{\gamma_{R}}\) for all \(f\in W\) determine a linear subspace whose
\(L^{2}\)-closure belongs to \(\Gr(r)\). This embedding automatically defines
a topology on \(\Grrat(r)\) and its subspaces by restriction.
\begin{lemma}
  \label{th:vdim-z-mc}
  A subspace \(W\in \oGrrat(r)\) has virtual dimension zero if and only if
  the codimension of the inclusion \(W\subseteq q^{-1}\mathcal{P}^{r}\)
  coincides with \(r\deg q\).
\end{lemma}
\begin{proof}
  Condition \eqref{def-oGrrat-m} may be rewritten as \(qp\mathcal{P}^{r}
  \subseteq qW\subseteq \mathcal{P}^{r}\), so that the codimension of \(W\) in
  \(q^{-1}\mathcal{P}^{r}\) is the same as the codimension of \(qW\) in
  \(\mathcal{P}^{r}\), namely \(\dim \quoz{\mathcal{P}^{r}}{qW}\). Taking the
  quotient of both those spaces by the common subspace \(qp\mathcal{P}^{r}\)
  we get an isomorphic linear space which is the quotient of two
  finite-dimensional spaces:
  \[ \frac{\mathcal{P}^{r}}{qW}\isom \frac{\quoz{\mathcal{P}^{r}}
    {qp\mathcal{P}^{r}}}{\quoz{qW}{qp\mathcal{P}^{r}}} \]
  Moreover, \(\quoz{qW}{qp\mathcal{P}^{r}}\isom \quoz{W}{p\mathcal{P}^{r}} =
  W'\), so that
  \[ \codim_{q^{-1}\mathcal{P}^{r}} W = r(\deg q + \deg p) - \dim W' \]
  By using \eqref{vdim-W} we finally get
  \begin{equation}
    \label{eq:codim-W}
    \codim_{q^{-1}\mathcal{P}^{r}} W = r \deg q - \vdim W
  \end{equation}
  from which the lemma follows.
\end{proof}
Let's introduce a more algebraic description for \(\oGrrat(r)\) in the same
vein as \cite{wils93}. For every \(k=1, \dots, r\), \(s\in \bb{N}\) and
\(\lambda\in \bb{C}\) we define the linear functional \(\eval_{k,s,\lambda}\)
on \(\mathcal{P}^{r}\) by
\[ \pair{\eval_{k,s,\lambda}}{(p_{1}, \dots, p_{r})} =
p_{k}^{(s)}(\lambda) \]
These functionals are easily seen to be linearly independent; we denote by
\(\mathscr{C}^{(r)}\) the linear space they generate, and think of it as a
space of ``differential conditions'' we can impose on \(r\)-tuples of
polynomials. We also set
\[ \mathscr{C}^{(r)}_{\lambda}\deq \linspan \{\eval_{k,s,\lambda}\}_{1\leq
  k\leq r,\, s\in \bb{N}} \]
and
\[ \mathscr{C}^{(r)}_{t,\lambda}\deq \linspan \{\eval_{k,s,\lambda}\}_{1\leq
  k\leq r,\, 0\leq s < t} \]
with the convention that \(\mathscr{C}^{(r)}_{0,\lambda} = \{0\}\).

Given \(c\in \mathscr{C}^{(r)}\), the finite set of points \(\lambda\in
\bb{C}\) such that the projection of \(c\) on \(\mathscr{C}^{(r)}_{\lambda}\)
is nonzero will be called the \emph{support} of \(c\). For every linear
subspace \(C\subseteq \mathscr{C}^{(r)}\), its annihilator
\[ V_{C}\deq \set{(p_{1}, \dots, p_{r})\in \mathcal{P}^{r} | \pair{c}{(p_{1},
    \dots, p_{r})} = 0 \text{ for all } c\in C} \]
is a linear subspace in \(\mathcal{P}^{r}\).
\begin{lemma}
  \label{th:dd-grrat-m}
  A subspace \(W\subseteq \mathcal{R}^{r}\) belongs to \(\oGrrat(r)\) if and
  only if there exist a finite-dimensional subspace \(C\subseteq
  \mathscr{C}^{(r)}\) and a polynomial \(q\) such that \(W = q^{-1}V_{C}\);
  moreover \(W\in \Grrat(r)\) if and only if \(r\deg q = \dim C\).
\end{lemma}
\begin{proof}
  Let's suppose that \(W = q^{-1}V_{C}\) for some finite-dimensional subspace
  \(C\) in \(\mathscr{C}^{(r)}\) and let \(\{\lambda_{1}, \dots,
  \lambda_{m}\}\) be the support of \(C\). For every \(i=1, \dots, m\) let
  \(t_{i}\) be the maximum value of \(s\) for the functionals
  \(\eval_{k,s,\lambda_{i}}\) involved in the elements of \(C\). By letting
  \(p\deq \prod_{i=1}^{m} (z-\lambda_{i})^{t_{i}+1}\) it follows that
  \(p\mathcal{P}^{r}\subseteq V_{C}\), hence a fortiori
  \(pq\mathcal{P}^{r}\subseteq V_{C}\subseteq \mathcal{P}^{r}\) and dividing
  by \(q\) we see that \(W\in \oGrrat(r)\). Now, if \(r\deg q = \dim C\) then
  \(\codim_{\mathcal{P}^{r}} V_{C} = \dim C = r\deg q\) and lemma
  \ref{th:vdim-z-mc} implies that \(W\in \Grrat(r)\).

  For the converse, let \(W\in \oGrrat(r)\). Then there exist \(p,q\in
  \mathcal{P}\) such that \(qp\mathcal{P}^{r}\subseteq qW\subseteq
  \mathcal{P}^{r}\); this means in particular that the linear space \(qW\) is
  obtained by imposing a certain (finite) number of linearly independent
  conditions in the dual space of
  \(\quoz{\mathcal{P}^{r}}{qp\mathcal{P}^{r}}\). The latter can be identified
  with \(\bb{C}^{r}\otimes U\), where \(U\) is the linear space of polynomials
  with degree less than \(\deg p + \deg q\), so that \(qW\) is determined by a
  finite-dimensional subspace in \((\bb{C}^{r}\otimes U)^{*}\). On the other
  hand, this space is generated by the elements of \(\mathscr{C}^{(r)}\)
  (e.g. using the functionals \(\frac{1}{s!} \eval_{k,s,0}\) that extract the
  \(s\)-th coefficient of the \(k\)-th polynomial). Thus, there exists a
  linear subspace \(C\subseteq \mathscr{C}^{(r)}\) of finite dimension such
  that \(V_{C} = qW\). If moreover \(\vdim W = 0\), then by \eqref{vdim-W} the
  linear space \(W'\isom \quoz{qW}{qp\mathcal{P}}\) has dimension \(r\deg p\),
  so that it must be defined by \(r\deg p\) linearly independent conditions.
  It follows that the subspace \(C\) has dimension \(r\deg q\).
\end{proof}
In the sequel, the subspace \(q^{-1}V_{C}\) singled out by this lemma will be
denoted simply by \((C,q)^{*}\).
\begin{lemma}
  \label{th:Gamma-r}
  Two subspaces \(W_{1} = (C,q_{1})^{*}\) and \(W_{2} = (C,q_{2})^{*}\)
  determined by the same conditions space \(C\) in \(\Grrat(r)\) lie in the
  same \(\Gamma_{-}(r)\)-orbit.
\end{lemma}
Actually, the matrix \(\eta\deq \frac{q_{1}}{q_{2}} I_{r}\) belongs to
\(\Gamma_{-}(r)\) (notice that \(q_{1}\) and \(q_{2}\) are both of degree
\(r\dim C\)) and is such that \(W_{1}\eta = W_{2}\).

We say that a finite-dimensional subspace \(C\subseteq \mathscr{C}^{(r)}\) is
\emph{homogeneous} if it admits a basis consisting of ``1-point conditions'',
i.e. differential conditions each one involving a single point:
\begin{equation}
  \label{eq:def-homC}
  C = \bigoplus_{\lambda\in \bb{C}} C_{\lambda} \quad\text{ where }\quad
  C_{\lambda}\deq C\cap \mathscr{C}^{(r)}_{\lambda}
\end{equation}
We denote by \(\Grhom(r)\) the set of subspaces \((C,q)^{*}\in \Grrat(r)\)
such that \(C\) is homogeneous, \(\dim C_{\lambda} = rn_{\lambda}\) for some
natural numbers \(n_{\lambda}\in \bb{N}\) and \(q=q_{C}\), where
\begin{equation}
  \label{eq:def-qC}
  q_{C}\deq \prod_{\lambda\in \bb{C}} (z-\lambda)^{n_{\lambda}}
\end{equation}

\subsection{The Baker and tau functions}

We would now like to determine the Baker function of a point \((C,q)^{*}\in
\Grrat(r)\). By Lemma \ref{th:Gamma-r} it is enough to consider the case
\(q=z^{d}\). Thus, we suppose that \(\dim C = rd\) and take the subspace \(W =
(C,z^{d})^{*}\); we claim that \(\psi_{W}\) must have the form
\begin{equation}
  \label{eq:fB-Grrat-m-1}
  \psi_{W}(g,z) = \biggl( I_{r} + \sum_{j=1}^{d} W_{j}(g)z^{-j}\biggr) g(z)
\end{equation}
Indeed, by standard arguments (cfr. \cite[Sect. 4]{wils93}), on the one hand
we have that each row of the matrix-valued function \(z^{d}\psi_{W}(g,z)\)
(for every fixed \(g\)) belongs to the \(L^{2}\)-closure of \(V_{C}\) in
\(H^{(r)}_{+}\) (for some value of the radius \(R\)), and on the other hand
that each functional \(\eval_{k,s,\lambda}\) extends uniquely to a continuous
functional on \(H^{(r)}_{+}\); hence \((z^{d}\psi_{W})_{\alpha}\in V_{C}\) for
every \(\alpha\). Now,
\[ (z^{d}\psi_{W})_{\alpha\beta} = \biggl( z^{d}\delta_{\alpha\beta} +
  \sum_{j\geq 1} W_{j\alpha\beta}(g) z^{d-j}\biggr) g_{\beta}(z) \]
and if we require every matrix element to be a polynomial, we see that every
\(W_{j}\) with \(j>d\) must be the zero matrix.

To determine the matrices \(W_{1}, \dots, W_{d}\), let \((c_{1}, \dots,
c_{rd})\) be a basis for \(C\); for each \(\alpha\) we take the \(\alpha\)-th
row of \(z^{d}\psi_{W}\) and impose the equalities \(\pair{c_{i}}{(z^{d}
  \psi_{W})_{\alpha}} = 0\) (with \(i=1,\dots,rd\)). This yields the following
linear system of equations:
\begin{equation}
  \label{eq:fB-Grrat-m-2}
  \pair{c_{i}}{\biggl( z^{d}\delta_{\alpha\beta} + \sum_{j=1}^{d}
    W_{j\alpha\beta}(g) z^{d-j}\biggr) g_{\beta}(z)} = 0
\end{equation}
In other words, we have a family of \(r\) linear systems, each of which
involves \(rd\) equations, for a total of \(r^{2}d\) scalar equations. The
unknowns are of course the \(r^{2}\) entries of the \(d\) matrices
\(\{W_{1}, \dots, W_{d}\}\); the coefficients of these unknowns involve, as
in the scalar case, the \(g_{\beta}\)'s and their derivatives evaluated at
the points in the support of \(C\).

The tau functions associated to \(W = (C,z^{d})^{*}\in \Grrat(r)\) are readily
obtained by imitating the calculations in \cite{dick93}; they are built out of
the following family of \((rd+1)\times (rd+1)\) matrices (indexed by
\(\alpha,\beta = 1, \dots, r\)):
\begin{equation}
  \label{eq:def-Mab}
  M_{\alpha\beta}\deq
  \begin{pmatrix}
    \pair{c_{1}}{g_{1}} & \hdots & \pair{c_{rd}}{g_{1}} & 0\\
    \vdots& & \vdots & z^{-d}\\
    \pair{c_{1}}{g_{r}} & \hdots & \pair{c_{rd}}{g_{r}} & 0\\
    \pair{c_{1}}{zg_{1}} & \hdots & \pair{c_{rd}}{zg_{1}} & 0\\
    \vdots& & \vdots & z^{1-d}\\
    \pair{c_{1}}{zg_{r}} & \hdots & \pair{c_{rd}}{zg_{r}} & 0\\
    \vdots&\vdots&\vdots&\vdots\\
    \pair{c_{1}}{z^{d-1}g_{1}} & \hdots & \pair{c_{rd}}{z^{d-1}g_{1}} & 0\\
    \vdots& & \vdots & z^{-1}\\
    \pair{c_{1}}{z^{d-1}g_{r}} & \hdots & \pair{c_{rd}}{z^{d-1}g_{r}} & 0\\
    \pair{c_{1}}{z^{d}g_{\alpha}} & \hdots & \pair{c_{rd}}{z^{d}g_{\alpha}}
    & \delta_{\alpha\beta}
  \end{pmatrix}
\end{equation}
where in the last column the only nonzero element is on the \(\beta\)-th row
of each block. Notice that in an expression like \(\pair{c_{i}}{g_{\gamma}}\),
\(g_{\gamma}\) must be interpreted as the row vector having \(g_{\gamma}\) in
its \(\gamma\)-th entry and zero elsewhere.

In terms of these matrices, the ``diagonal'' tau function \(\tau_{W}\) is
simply the cofactor of the element \(\delta_{\alpha\beta} = 1\) in the lower
right corner, or equivalently the determinant of the \(rd\times rd\) minor
obtained by deleting the last column and the last row:
\begin{equation}
  \label{eq:24}
  \tau_{W}(g) = \det (\pair{c_{i}}{z^{j-1}g_{\gamma}})_{\substack{
      i=1\dots rd\\
      j=1\dots d,\, \gamma=1\dots r}}
\end{equation}
This is just the matrix of coefficients of the linear system
\eqref{fB-Grrat-m-2}, so that the system has a solution exactly when \(g\in
\Gamma_{+}(r)^{W}\), as expected. The ``off-diagonal'' tau function
\(\tau_{W\alpha\beta}\) (with \(\alpha\neq \beta\)) is the cofactor of the
element \(z^{-1}\) in the last column:
\begin{equation}
  \label{eq:ftauab-Grrat}
  \tau_{W\alpha\beta}(g) = (-1)^{r-\beta} \det (\pair{c_{i}}{z^{j-1}
    (g_{\gamma} + \delta_{jd}\delta_{\gamma\beta} (zg_{\alpha} -
    g_{\gamma}))})_{\substack{
      i=1\dots rd\\
      j=1\dots d,\, \gamma=1\dots r}}
\end{equation}
Observe that an expression such as \(z^{j-1} g_{\gamma}\) may equivalently be
read as \(\partial_{1\gamma}^{j-1} g_{\gamma}\); this will be useful in what
follows.

\section{Rational solutions}
\label{sec:4}

We can now prove a rationality result for the solutions of the mcKP hierarchy
coming from subspaces in \(\Grhom(r)\) defined by a set of condition whose
support is a single point.
\begin{theorem}
  \label{th:mcKP}
  Let \(W\in \Grhom(r)\) and suppose that \(W = (C,q_{C})^{*}\) with \(C\)
  supported on a single point \(\lambda\in \bb{C}\). Then:
  \begin{enumerate}
  \item Each diagonal entry of the reduced Baker function \(\tilde{\psi}_{W}\)
    is a rational function of the times \(t_{1}^{(1)}, \dots, t_{1}^{(r)}\)
    that tends to \(1\) as \(t_{1}^{(\alpha)}\rightarrow \infty\) for any
    \(\alpha\);
  \item The tau function \(\tau_{W}\) is a polynomial in \(t_{1}^{(1)}, \dots,
    t_{1}^{(r)}\) with constant leading coefficient.
  % \item \(\tilde{\psi}_{W\alpha\alpha}\) is a rational function of
  %   \(t_{1}^{(1)}, \dots, t_{1}^{(r)}\) which tends to \(1\) as
  %   \(t_{1}^{(\alpha)}\rightarrow \infty\) for all \(\alpha\);
  % \item \(\tilde{\psi}_{W\alpha\beta}\) (for \(\alpha\neq \beta\)) is equal to
  %   \(\frac{g_{\alpha}(\lambda)}{g_{\beta}(\lambda)}\) times a rational function
  %   of \(t_{1}^{(1)}, \dots, t_{1}^{(r)}\) which tends to \(0\) as
  %   \(t_{1}^{(\alpha)}\rightarrow \infty\) for all \(\alpha\);
  % \item \(\tau_{W}\) is a polynomial in \(t_{1}^{(1)}, \dots, t_{1}^{(r)}\)
  %   with constant leading coefficient;
  % \item \(\tau_{W\alpha\beta}\) (for \(\alpha\neq \beta\)) is equal to
  %   \(\frac{g_{\alpha}(\lambda)}{g_{\beta}(\lambda)}\) times a polynomial in
  %   \(t_{1}^{(1)}, \dots, t_{1}^{(r)}\) with constant leading coefficient.
  \end{enumerate}
\end{theorem}
\begin{proof}
  The two statements are equivalent by virtue of the diagonal part of Sato's
  formula \eqref{15}, so it suffices to prove one of them; we choose the
  first. By hypothesis we have \(W = (C,(z-\lambda)^{d})^{*}\) with
  \(C\subseteq \mathscr{C}^{(r)}_{\lambda}\) of dimension \(rd\); let
  \((c_{1}, \dots, c_{rd})\) be a basis for it. Consider the subspace \(U\deq
  (C,z^{d})^{*}\in \Grrat(r)\); its tau function is given by \eqref{24}. To
  compute the diagonal elements of the corresponding Baker function we use
  Sato's formula for \(U\):
  \begin{equation}
    \label{eq:3:t-Sato}
    \tilde{\psi}_{U\alpha\alpha} = \frac{\tau_{U}(gq_{\zeta\alpha})}{\tau_{U}(g)}
  \end{equation}
  (here we consider \(\zeta\) as a parameter and \(z\) as a variable). Since
  we are only interested in the times with subscript \(1\) we will work in the
  stationary setting, i.e. we put \(t_{k}^{(\alpha)}=0\) for every \(k\geq
  2\), \(\alpha=1\dots m\). Each condition \(c_{i}\) is supported at
  \(\lambda\), hence we can define a family of polynomials
  \(\{\phi_{i\gamma}\}_{i=1\dots rd,\, \gamma=1\dots r}\) (with
  \(\phi_{i\gamma}\) only depending on \(t_{1}^{(\gamma)}\)) by the equation
  \begin{equation}
    \label{eq:3:def-phi}
    \pair{c_{i}}{g_{\gamma}} = g_{\gamma}(\lambda)\phi_{i\gamma}
  \end{equation}
  To apply \eqref{3:t-Sato} we need to know the determinant of the matrices
  \(\partial_{1\gamma}^{j-1} \pair{c_{i}}{g_{\gamma}}\) and
  \(\partial_{1\gamma}^{j-1} \pair{c_{i}}{g_{\gamma}(1 - \delta_{\alpha\gamma}
    \frac{z}{\zeta})}\). As for the first, by using \eqref{3:def-phi} its
  generic element can be written in the form
  \begin{equation}
    \label{eq:3:tmp-1}
    \partial_{1\gamma}^{j-1} \pair{c_{i}}{g_{\gamma}} = \partial_{1\gamma}^{j-1}
    (g_{\gamma}(\lambda) \phi_{i\gamma}) = g_{\gamma}(\lambda)
    (\partial_{1\gamma} + \lambda)^{j-1} \phi_{i\gamma}
  \end{equation}
  For the second matrix we have
  \begin{equation}
    \label{eq:3:tmp-2}
    \partial_{1\gamma}^{j-1} \pair{c_{i}}{g_{\gamma}(1 - \delta_{\alpha\gamma}
    \frac{z}{\zeta})} = \partial_{1\gamma}^{j-1} \pair{c_{i}}{g_{\gamma}}
    - \partial_{1\gamma}^{j-1} \pair{c_{i}}{\delta_{\alpha\gamma} g_{\gamma}
      \frac{z}{\zeta}}
  \end{equation}
  The first term is exactly \eqref{3:tmp-1}, whereas the second is
  \[ \delta_{\alpha\gamma} \partial_{1\gamma}^{j} \pair{c_{i}}{g_{\gamma}}
  \zeta^{-1} = \delta_{\alpha\gamma} \partial_{1\gamma}^{j}
  (g_{\gamma}(\lambda)\phi_{i\gamma}) \zeta^{-1} = 
  \delta_{\alpha\gamma} g_{\gamma}(\lambda) (\partial_{1\gamma} +
  \lambda)^{j} \phi_{i\gamma} \zeta^{-1} \]
  By putting all together, equation \eqref{3:tmp-2} becomes
  \[ g_{\gamma}(\lambda) (\partial_{1\gamma} + \lambda)^{j-1} \bigl(
    \phi_{i\gamma} - \delta_{\alpha\gamma} (\partial_{1\gamma}
    \phi_{i\gamma} + \lambda \phi_{i\gamma}) \zeta^{-1}\bigr) \]
  that we can rewrite as
  \begin{equation}
    \label{eq:3:tmp-3}
    g_{\gamma}(\lambda) (1 - \delta_{\alpha\gamma} \frac{\lambda}{\zeta})
    (\partial_{1\gamma} + \lambda)^{j-1} \bigl( \phi_{i\gamma} - 
    \delta_{\alpha\gamma} \frac{1}{\zeta-\lambda} \partial_{1\gamma}
    \phi_{i\gamma}\bigr)
  \end{equation}
  But \((1 - \delta_{\alpha\gamma} \frac{\lambda}{\zeta}) =
  q_{\zeta\alpha}(\lambda)_{\gamma}\), so plugging \eqref{3:tmp-1} and
  \eqref{3:tmp-3} into \eqref{3:t-Sato} we obtain
  \begin{equation}
    \label{eq:3:fin}
    \tilde{\psi}_{U\alpha\alpha}(g,\zeta) = (q_{\zeta}(\lambda))^{d}
    \frac{\det \left( (\partial_{1\gamma} + \lambda)^{j-1} (\phi_{i\gamma} -
        \delta_{\alpha\gamma} \frac{1}{\zeta-\lambda} \partial_{1\gamma}
        \phi_{i\gamma})\right)}{\det \left( (\partial_{1\gamma} +
        \lambda)^{j-1} \phi_{i\gamma}\right)}
  \end{equation}
  The factor \((q_{\zeta}(\lambda))^{d}\) disappears when we go back from
  \(\psi_{U}\) to \(\psi_{W}\); we are left with the ratio of two determinants
  of matrices with polynomial entries in the times \(t_{1}^{(\gamma)}\), which
  is clearly a rational function. Moreover, if we expand the numerator of
  \eqref{3:fin} by linearity over the sum, we see that the term obtained by
  always choosing \(\phi_{i\gamma}\) exactly reproduces the polynomial at the
  denominator, and all the other terms involve a polynomial which has degree
  strictly lower than \(\tau\) in some \(t_{1}^{(\gamma)}\) (since we replace
  \(\phi_{i\gamma}\) with one of its derivatives); this proves that
  \(\tilde{\psi}_{W\alpha\alpha}\rightarrow 1\) as all the
  \(t_{1}^{(\alpha)}\) tend to infinity.
\end{proof}
We now consider the evolution on \(\Gr(r)\) described by the flows of the
matrix KP hierarchy. Recall from the Introduction that these are given by the
vector fields \(\partial/\partial t_{k}\), where \(t_{k} = \sum_{\gamma}
t_{k}^{(\gamma)}\). Equivalently, we can take the matrix \(g\in \Gamma_{+}(r)\)
to be of the form
\begin{equation}
  \label{eq:3:red-mKP}
  g = \diag(\e^{\xi(r^{-1}\bm{t},z)}, \dots, \e^{\xi(r^{-1}\bm{t},z)})
\end{equation}
where \(\bm{t} = (t_{k})_{k\geq 1}\) (and \(t_{1} = x\)). Let's define
\(\tilde{g}\deq \e^{\xi(r^{-1}\bm{t},z)}\) and \(h\deq \e^{\xi(\bm{t},z)}\)
(so that \(g = \tilde{g} I_{r}\) and \(\tilde{g}^{r} = h\)); then the Baker
and tau functions for the matrix KP hierarchy are naturally expressed in terms
of \(h\) only, since that single function completely controls the flows of the
hierarchy.
\begin{theorem}
  \label{th:mKP}
  Let \(W\in \Grhom(r)\), then:
  \begin{enumerate}
  \item \(\tilde{\psi}_{W}(h,z)\) is a matrix-valued rational function of
    \(x\) that tends to \(I_{r}\) as \(x\rightarrow \infty\);
  \item \(\tau_{W}(h)\) and \(\tau_{W\alpha\beta}(h)\) are polynomial
    functions of \(x\) with constant leading coefficients.
  \end{enumerate}
\end{theorem}
\begin{proof}
  Let \(W = (C,q_{C})^{*}\) with \(C\) homogeneous and let \((\lambda_{1},
  \dots, \lambda_{d})\) be its support with each point counted according to
  its multiplicity (so that the \(\lambda_{i}\) are not necessarily distinct);
  finally for each \(i=1\dots d\) let \((c_{ij})_{j=1\dots r}\) be a set of
  \(r\) linearly independent conditions at \(\lambda_{i}\). In this way we get
  a basis \((c_{11}, \dots, c_{dr})\) for \(C\) made of 1-point conditions.
  Now consider the subspace \(U\deq (C,z^{d})^{*}\in \Grrat(r)\); it is
  related to \(W\) by the following element of \(\Gamma_{-}(r)\):
  \begin{equation}
    \label{eq:3:eta-mKP}
    \eta = \prod_{i=1}^{d} q_{z}(\lambda_{i})^{-1} I_{r} = \exp \biggl(
    \sum_{k\geq 0} \sum_{i=1}^{d} \frac{\lambda_{i}^{k}}{k} z^{-k}\biggr)
    I_{r}
  \end{equation}
  This corresponds %(cfr. \cite[Lemma 3.8]{sw85})
  to multiplying the tau
  function by
  \begin{equation}
    \label{eq:3:def-hatg}
    \hat{\eta} = \prod_{\alpha=1}^{r} \exp \biggl( -\sum_{k\geq 0}
    \sum_{i=1}^{d} \lambda_{i}^{k} t^{(\alpha)}_{k}\biggr) =
    \prod_{\alpha=1}^{r} \prod_{i=1}^{d} g_{\alpha}(\lambda_{i})^{-1} =
    \prod_{i=1}^{d} h(\lambda_{i})^{-1}
  \end{equation}
  since \(g_{\alpha} = \tilde{g}\) for every \(\alpha\) and \(\tilde{g}^{r} =
  h\).

  Let's define the family of polynomials \(\{\phi_{ij\gamma}\}\) (for
  \(i=1\dots d\), \(j=1\dots r\), \(\gamma=1\dots r\)) by the equation
  \begin{equation}
    \label{eq:3:def-phi-ijb}
    \pair{c_{ij}}{g_{\gamma}} = \tilde{g}(\lambda_{i}) \phi_{ij\gamma}
  \end{equation}
  Again, this works precisely because each \(c_{ij}\) is a 1-point
  condition; notice that, although \(g_{\gamma} = \tilde{g}\) for every
  \(\gamma\), the polynomials \(\phi\) still depend on \(\gamma\) since it
  is the index of the only nonzero entry of the row vector on which
  \(c_{ij}\) acts. 
  % Now let \(\sigma_{ij}\) be the order of the highest
  % derivative involved in the condition \(c_{ij}\) (for any component); then
  % we claim that each \(\phi_{ij\gamma}\) has order \(\sigma_{ij}\) in
  % \(t_{1}^{(\gamma)}\) (or equivalently in \(x\)). Indeed (see the footnote
  % at page \pageref{fn:FdB}) the action of a \(k\)-th derivative on
  % \(\tilde{g}(z)\) produces the \(k\)-th Faà di Bruno differential
  % polynomial \(P_{k}(\xi')\), where \(\xi'\propto x + 2t_{2}z + \dots\), and
  % such an expression has leading term \(x^{k}\) in \(x\); hence \(c_{ij}\)
  % produces by definition a polynomial with leading term \(x^{\sigma_{ij}}\).

  We can now easily compute the tau functions associated to \(U\) by
  retracing the same steps as in the proof of theorem \ref{th:mcKP}, but
  now using the polynomials defined by \eqref{3:def-phi-ijb}. Then
  \(\tau_{U}(g)\) is the determinant of the matrix whose generic element is
  \[ \tilde{g}(\lambda_{i}) (\partial_{1\gamma} + \lambda_{i})^{k-1}
  \phi_{ij\gamma} \]
  The term \(\tilde{g}(\lambda_{i})\) does not depend on the row indices
  \((k,\gamma)\) so that we can factor it out from the determinant and get
  \begin{equation}
    \label{eq:3:tau-U-mKP}
    \tau_{U} = \prod_{i=1}^{d} (\tilde{g}(\lambda_{i}))^{r} \det
    \left( (\partial_{1\gamma} + \lambda_{i})^{k-1} \phi_{ij\gamma}\right)
  \end{equation}
  But \(\prod_{i} (\tilde{g}(\lambda_{i}))^{r} = \prod_{i} h(\lambda_{i})\)
  is exactly the inverse of \eqref{3:def-hatg}, so that
  \begin{equation}
    \label{eq:3:tau-W-mKP}
    \tau_{W} = \det \left( (\partial_{1\gamma} + \lambda_{i})^{k-1}
      \phi_{ij\gamma}\right)
  \end{equation}
  is the determinant of a matrix with polynomial entries in
  \(t_{1}^{(\alpha)} = \frac{x}{r}\), hence a polynomial in \(x\) and the
  coefficient of the top degree term involves only the constants
  \(\lambda_{i}\). By Sato's formula, this implies that
  \(\tilde{\psi}_{W\alpha\alpha}\rightarrow 1\) as \(x\rightarrow \infty\).

  Now take \(\alpha,\beta\in \{1, \dots, m\}\), \(\alpha\neq \beta\) and
  consider the off-diagonal tau function \(\tau_{U\alpha\beta}(g)\); it is
  given by the determinant of a matrix \(M_{\alpha\beta}(g)\) which
  coincides with the one involved in the definition of \(\tau_{U}(g)\)
  except for the row corresponding to \(k=d-1\), \(\gamma=\beta\) which is
  replaced by the row \(\pair{c_{ij}}{\partial_{1\alpha}^{d} g_{\alpha}}\).
  But since \(g_{\alpha} = g_{\beta} = \tilde{g}\) we can again collect out
  of the determinant the same factor as before, so that
  \(\tau_{W\alpha\beta}(g) = \det \Phi_{ij,k\gamma}\) with
  \begin{equation}
    \label{eq:def-Phi-ijkg}
    \Phi_{ij,k\gamma}\deq
    \begin{cases}
      (\partial_{1\gamma} + \lambda_{i})^{k-1} \phi_{ij\gamma} &\text{ if }
      k\neq d\text{ or } \gamma\neq\beta\\
      (\partial_{1\alpha} + \lambda_{i})^{d} \phi_{ij\alpha} &\text{ if }
      k=d\text{ and } \gamma=\beta
    \end{cases}
  \end{equation}
  This is also a polynomial in \(x\); moreover we can write
  \begin{equation}
    \label{eq:3:subst-row}
    (\partial_{1\alpha} + \lambda_{i})^{d} \phi_{ij\alpha} =
    (\partial_{1\alpha} + \lambda_{i})^{d-1} (\lambda_{i}\phi_{ij\alpha}
    + \partial_{1\alpha}\phi_{ij\alpha})
  \end{equation}
  But now \(M_{\alpha\beta}\) has also a row (for \(k=d-1\),
  \(\gamma=\alpha\)) whose generic  element reads \((\lambda_{i}
  + \partial_{1\alpha})^{d-1} \phi_{ij\alpha}\), and we can subtract this
  row multiplied by \(\lambda_{1}\) (say) to the row \eqref{3:subst-row}
  without altering the determinant, so that
  \[ \Phi_{ij,k\gamma}\deq
  \begin{cases}
    (\partial_{1\gamma} + \lambda_{i})^{k-1} \phi_{ij\gamma} &\text{ if }
    k\neq d\text{ or } \gamma\neq\beta\\
    (\partial_{1\alpha} + \lambda_{i})^{d-1} ((\lambda_{i}-\lambda_{1})
    \phi_{ij\alpha} + \partial_{1\alpha}\phi_{ij\alpha}) &\text{ if }
    k=d\text{ and } \gamma=\beta
  \end{cases} \]
  This means that \(\tau_{W\alpha\beta}\) is the determinant of a matrix
  whose generic entry is equal or of degree strictly lower than the
  corresponding one on \(\tau_{W}\); it follows that the degree of
  \(\tau_{W\alpha\beta}\) is strictly lower than \(\tau_{W}\), and this
  (again by Sato's formula) implies that the off-diagonal components of
  \(\tilde{\psi}_{W}\) tend to zero as \(x\rightarrow \infty\).
\end{proof}
To see how some of the new solutions look like, take for example the
homogeneous subspace with 1-point support \(W = (C,\frac{1}{z-\lambda})^{*}\)
with \(C\) generated by the \(r\) conditions
\begin{equation}
  \label{eq:cond-sp}
  c_{k} = \eval_{k,1,\lambda} + a_{k1} \eval_{1,0,\lambda} + \dots + a_{kr}
  \eval_{r,0,\lambda} \qquad (k=1\dots r)
\end{equation}
determined by the \(r\times r\) matrix \(A = (a_{ij})\). These are, in a
sense, the most general conditions involving only the functionals
\(\eval_{k,s,\lambda}\) with \(s\leq 1\). For the sake of brevity, let's set
\(\bm{t}_{\lambda}\deq r^{-1}(x + 2t_{2}\lambda + 3t_{3}\lambda^{2} +
\dots)\); then
\[ \tau_{W}(\bm{t}) = \det (\bm{t}_{\lambda}I_{r} + \trasp{A}) \qquad
\tau_{W\alpha\beta}(\bm{t}) = -\cof_{\beta,\alpha} (\bm{t}_{\lambda}I_{r} +
\trasp{A}) \]
\[ \tilde{\psi}(\bm{t},z) = I_{r} - (\bm{t}_{\lambda}I_{r} + \trasp{A})^{-1}
\frac{1}{z-\lambda} \]
where \(\cof_{\beta,\alpha}\) stands for the \((\beta,\alpha)\)-cofactor of a
matrix and \(\trasp{A}\) for the transpose of \(A\).

More generally, we could take the element of \(\Grhom(r)\) determined by the
support \(\{\lambda_{1}, \dots, \lambda_{n}\}\) and, for each \(i=1\dots n\),
a matrix \(A_{i}\) that specifies a set of conditions of the form
\eqref{cond-sp} to be imposed at the point \(\lambda_{i}\). The tau function
of such a subspace is the determinant of the following block matrix:
\[
\begin{pmatrix}
  Y_{1} & \hdots & Y_{n}\\
  \lambda_{1}Y_{1} + I_{r} & \hdots & \lambda_{n}Y_{n} + I_{r}\\
  \vdots & \ddots & \vdots\\
  \lambda_{1}^{n-1} Y_{1} + (n-1)\lambda_{1}^{n-2} I_{r} & \hdots &
  \lambda_{n}^{n-1} Y_{n} + (n-1)\lambda_{n}^{n-2} I_{r}
\end{pmatrix}
\]
where \(Y_{i}\deq \bm{t}_{\lambda_{i}}I_{r} + \trasp{A_{i}}\). The off-diagonal
tau functions \(\tau_{W\alpha\beta}\) are obtained in the usual way (i.e.,
replacing the \(\beta\)-th line of the bottom blocks with the \(\alpha\)-th
line of the blocks \(\lambda_{i}^{n} Y_{i} + n\lambda_{i}^{n-1} I_{r}\)), and
the matrix Baker function is then given by Sato's formula.  When \(r=1\),
these subspaces are exactly the ones described in \cite[Sect. 3]{wils98}.

\section{Relationship with the multicomponent KP/CM correspondence}
\label{sec:5}

Recall from \cite{wils98} that a crucial ingredient of the scalar KP/CM
correspondence is the bijective map \(\app{\beta}{\CM}{\Grad}\) between the
phase space of the Calogero-Moser system \(\CM\) and the adelic Grassmannian
\(\Grad\). In the multicomponent case, an analogous mapping should be defined
between the phase space of the Gibbons-Hermsen system \cite{gh84} (also known
as ``spin Calogero-Moser'') and some space of solutions for the matrix KP
hierarchy, as suggested by a well-known calculation \cite{bbkt94}.

In Wilson's unpublished notes \cite{wils09}, a map that fulfils this r\^ole is
conjectured. In order to describe (part of) the definition of this map let's
recall (see e.g. \cite[Sect. 8]{wils98}) that the completed phase space of
the \(n\)-particle, \(r\)-component Gibbons-Hermsen system is the (smooth,
irreducible) affine algebraic variety defined by the following symplectic
reduction:
\[ \mathcal{C}_{n,r} = \quoz{\set{(X,Y,v,w) | \comm{X}{Y} + vw =
    -I}}{\GL(n,\bb{C})} \]
where \((X,Y,v,w)\in \End(\bb{C}^{n}) \oplus \End(\bb{C}^{n}) \oplus
\Hom(\bb{C}^{n},\bb{C}^{r}) \oplus \Hom(\bb{C}^{r},\bb{C}^{n})\) and
\(\GL(n,\bb{C})\) acts as a change of basis in \(\bb{C}^{n}\). Now denote by
\(\CM_{n,r}''\) the subspace of \(\CM_{n,r}\) consisting of equivalence
classes of quadruples for which \(Y\) is diagonalisable with distinct
eigenvalues; from each of these classes we can select a representative such
that \(Y = \diag (\lambda_{1}, \dots, \lambda_{n})\), \(X\) is a Moser matrix
associated to \(Y\) (i.e. \(X_{ij} = 1/(\lambda_{i}-\lambda_{j})\) and
\(X_{ii} = \alpha_{i}\) is any complex number) and each pair \((v_{i},w_{i})\)
(where we denote by \(v_{i}\) the \(i\)-th row of \(v\) and by \(w_{i}\) the
\(i\)-th column of \(w\)) belongs to the algebraic variety
\(\set{(\xi,\eta)\in \bb{C}^{r}\times \bb{C}^{r} | \xi\cdot \eta =
  -1}/\bb{C}^{*}\), where \(\lambda\in \bb{C}^{*}\) acts as
\begin{equation}
  \label{eq:16}
  \lambda.(\xi,\eta) = (\lambda\xi, \lambda^{-1}\eta)
\end{equation}
Notice in particular that none of the \(v_{i}\)'s and \(w_{i}\)'s can be the
zero vector, by virtue of the normalisation condition \(v_{i}w_{i} = -1\).

According to \cite{wils09}, a point \([X,Y,v,w]\in \CM_{n,r}''\) corresponds
to the subspace \(W\in \Grrat(r)\) defined by the following prescriptions:
\begin{enumerate}
\item \label{e:Wcond-reg} functions in \(W\) are regular except for (at most)
  a simple pole in each \(\lambda_{i}\) and a pole of any order at infinity;
\item \label{e:Wcond-lau} if \(f = \sum_{k\geq -1} f_{k}^{(i)}
  (z-\lambda_{i})^{k}\) is the Laurent expansion of \(f\in W\) in
  \(\lambda_{i}\) then:
  \begin{enumerate}
  \item \label{e:Wcond-1} \(f_{-1}^{(i)}\) is a scalar multiple of \(v_{i}\),
    and
  \item \label{e:Wcond-2} \((f_{0}^{(i)} + \alpha_{i} f_{-1}^{(i)})\cdot w_{i}
    = 0\).
  \end{enumerate}
\end{enumerate}
Our purpose is now to show how these prescriptions translate in the language
of the previous sections.

So suppose that \(W\in \Grrat(r)\) satisfies conditions
(\ref{e:Wcond-reg}--\ref{e:Wcond-lau}); we must find a finite-dimensional,
homogeneous space of conditions \(C\) such that \(W = (C,q_{C})^{*}\) where
\(q_{C}\) is given by \eqref{def-qC}. Condition \ref{e:Wcond-reg} clearly
implies that \(C\) has support \(\{\lambda_{1}, \dots, \lambda_{n}\}\) and
\(q_{C} = \prod_{i=1}^{n} (z-\lambda_{i})\). Thus, we only need to find, for
each \(i=1\dots n\), a subspace \(C_{i}\subseteq
\mathscr{C}^{(r)}_{\lambda_{i}}\) of dimension \(r\) whose elements satisfy
the conditions (\ref{e:Wcond-1}--\ref{e:Wcond-2}). Let's define the \(n\)
polynomials
\[ q_{i}\deq \frac{q_{C}}{z-\lambda_{i}} = \prod_{j\neq i} (z-\lambda_{j}) \]
for \(i=1\dots n\); then by a direct computation we have that
\[ \eval_{k,0,\lambda_{i}}(q_{C}f) = q_{i}(\lambda_{i}) f^{(i)}_{-1,k} \]
and similarly
\[ \eval_{k,1,\lambda_{i}}(q_{C}f) = \sum_{\ell\neq i} \prod_{j\neq i,\ell}
(\lambda_{i}-\lambda_{j}) f^{(i)}_{-1,k} + q_{i}(\lambda_{i}) f^{(i)}_{0,k} \]
whence (putting \(\delta_{i}\deq \sum_{\ell\neq i}
(\lambda_{i}-\lambda_{\ell})^{-1}\))
\[ q_{i}(\lambda_{i}) f^{(i)}_{0,k} = (\eval_{k,1,\lambda_{i}} - \delta_{i}
\eval_{k,0,\lambda_{i}}) (q_{C}f) \]
Using these expressions, (\ref{e:Wcond-1}) reads
\[ \eval_{k,0,\lambda_{i}}(q_{C}f) = q_{i}(\lambda_{i}) a_{i} v_{ik}
\quad\text{ for all } k=1\dots r \]
for some \(a_{i}\in \bb{C}\), but only \(r-1\) of these conditions are
independent because of the \(\bb{C}^{*}\)-action \eqref{16} on \(v_{i}\). As
we already noticed, at least one entry of \(v_{i}\) is nonzero, say
\(v_{i1}\neq 0\). Then we can use the \(\bb{C}^{*}\)-action to normalise
\(v_{i1} = 1\), so that \(\eval_{1,0,\lambda_{i}} (q_{C}f) =
q_{i}(\lambda_{i}) a_{i}\) and the remaining \(r-1\) conditions can be written
as
\[ \eval_{k,0,\lambda_{i}}(q_{C}f) - v_{ik} \eval_{1,0,\lambda_{i}}(q_{C}f) =
0 \quad\text{ for all } k=2\dots r \]
The further condition (\ref{e:Wcond-2}) directly translates as
\[ \sum_{k=1}^{r} w_{ki} \left( \eval_{k,1,\lambda_{i}} (q_{C}f) + (\alpha_{i} -
  \delta_{i}) \eval_{k,0,\lambda_{i}} (q_{C}f)\right) = 0 \]
where \(w_{1i}\) is fixed by the equation \(v_{i}w_{i} = -1\). We conclude
that, in the \((r-1)\)-dimensional affine cell where \(v_{i1}\neq 0\), we can
take as \(C_{i}\) the subspace
\[ C_{i} = \left\langle \eval_{2,0,\lambda_{i}} - v_{i2} \eval_{1,0,\lambda_{i}},
\dots, \eval_{r,0,\lambda_{i}} - v_{ir} \eval_{1,0,\lambda_{i}},
\sum_{k=1}^{r} w_{ki} (\eval_{k,1,\lambda_{i}} + (\alpha_{i} - \delta_{i})
\eval_{k,0,\lambda_{i}}) \right\rangle \]
Analogous descriptions are available in the other cells, where \(v_{ik}=0\)
for every \(k<k^{*}\) and \(v_{ik^{*}}\neq 0\). Repeating this argument for
every \(i=1\dots n\), we obtain a homogeneous subspace \(C = C_{1}\oplus
\dots\oplus C_{n}\) in \(\mathscr{C}^{(r)}\) such that \(W = (C,q_{C})^{*}\),
as we wanted. 

To see a concrete example in the simplest possible case (\(n=r=2\)), consider
the space associated to a point
\[ \left[
\begin{pmatrix}
  \alpha_{1} & \frac{v_{1}w_{2}}{\lambda_{1}-\lambda_{2}}\\
  \frac{v_{2}w_{1}}{\lambda_{2}-\lambda_{1}} & \alpha_{2}
\end{pmatrix},
\begin{pmatrix}
  \lambda_{1} & 0\\
  0 & \lambda_{2}
\end{pmatrix},
\begin{pmatrix}
  v_{11} & v_{12}\\
  v_{21} & v_{22}
\end{pmatrix},
\begin{pmatrix}
  w_{11} & w_{12}\\
  w_{21} & w_{22}
\end{pmatrix}
\right] \in \mathcal{C}_{2,2}'' \]
Suppose further that \(v_{11}\neq 0\) and \(v_{21}\neq 0\); the corresponding
subspace \(W\in \Grhom(2)\) is then given by \((z-\lambda_{1})^{-1}
(z-\lambda_{2})^{-1} V_{C}\), where \(C\subseteq \mathscr{C}^{(2)}\) is
generated by the four conditions
\begin{align*}
  c_{11} &= \eval_{2,0,\lambda_{1}} - v_{12} \eval_{1,0,\lambda_{1}}\\
  c_{12} &= w_{11} \eval_{1,1,\lambda_{1}} + w_{21} \eval_{2,1,\lambda_{1}} +
  w_{11} (\alpha_{1}+\delta) \eval_{1,0,\lambda_{1}} + w_{21}
  (\alpha_{1}+\delta) \eval_{2,0,\lambda_{1}}\\
  c_{21} &= \eval_{2,0,\lambda_{2}} - v_{22} \eval_{1,0,\lambda_{2}}\\
  c_{22} &= w_{12} \eval_{1,1,\lambda_{2}} + w_{22} \eval_{2,1,\lambda_{2}} +
  w_{12} (\alpha_{2}-\delta) \eval_{1,0,\lambda_{2}} + w_{22}
  (\alpha_{2}-\delta) \eval_{2,0,\lambda_{2}}
\end{align*}
where \(\delta\deq \delta_{2}=-\delta_{1}=(\lambda_{2}-\lambda_{1})^{-1}\),
\(w_{11}=-1-v_{12}w_{21}\) and \(w_{12}=-1-v_{22}w_{22}\). After some
computation we find
\[ \tau_{W}(\bm{t}) = (\bm{t}_{\lambda_{1}} + \alpha_{1})(\bm{t}_{\lambda_{2}}
+ \alpha_{2}) + \delta^{2} (v_{22}w_{21} + w_{11})(v_{12}w_{22} + w_{12}) \]
\[ \tau_{W12}(\bm{t}) = -w_{12}v_{22} (\bm{t}_{\lambda_{1}} + \alpha_{1}) -
w_{11}v_{12} (\bm{t}_{\lambda_{2}} + \alpha_{2}) + \delta \bigl( 2v_{12}v_{22}
(w_{12}w_{21}-w_{11}w_{22}) + v_{22}w_{12} - v_{12}w_{11}\bigr) \]
\[ \tau_{W21}(\bm{t}) = -w_{22}(\bm{t}_{\lambda_{1}} + \alpha_{1}) -
w_{21}(\bm{t}_{\lambda_{2}} + \alpha_{2}) + \delta \bigl( 2w_{21}w_{22}
(v_{22}-v_{12}) + w_{21} - w_{22}\bigr) \]
We remark that \(\tau_{W}\) coincides with the determinant of the matrix \(X\)
under the evolution described by the Gibbons-Hermsen flows \(H_{k} = \tr
Y^{k}\) (taking \(t_{k}\) as the time associated to \(H_{k}\)):
\[ \tau_{W}(\bm{t}) = \det
\begin{pmatrix}
  \alpha_{1} + \bm{t}_{\lambda_{1}} &
  \frac{\tilde{g}(\lambda_{1})}{\tilde{g}(\lambda_{2})}
  \frac{v_{1}w_{2}}{\lambda_{1} - \lambda_{2}}\\
  \frac{\tilde{g}(\lambda_{2})}{\tilde{g}(\lambda_{1})}
  \frac{v_{2}w_{1}}{\lambda_{2} - \lambda_{1}} & \alpha_{2} + \bm{t}_{\lambda_{2}}
\end{pmatrix} \]
This result fits nicely with the general formulae derived in \cite{wils09}.

\bibliographystyle{plain}
\bibliography{master}

\end{document}